\documentclass[conference]{IEEEtran}
\IEEEoverridecommandlockouts
\usepackage{cite}
\usepackage{amsmath,amssymb,amsfonts}
\usepackage{amsthm}
\usepackage{algorithm}
\usepackage{algorithmic}
\usepackage{graphicx}
\usepackage{textcomp}
\usepackage{xcolor}
\usepackage{booktabs}
\usepackage{multirow}
\usepackage{pifont}
\usepackage{hyperref}
\hypersetup{
  hidelinks,
  pdftitle={Identifying the Group to Intervene on to Maximize Effect Under Cross-Group Interference},
  pdfauthor={Xiaojing Du, Jiuyong Li, Lin Liu, Debo Cheng, Jixue Liu, Thuc Duy Le},
  pdfsubject={IEEE International Conference on Data Mining (ICDM) 2026},
  pdfkeywords={Causal Inference; Network Interference; Cross-Group Spillover Effects; Causal Effect Maximization; Subset Intervention}}
\usepackage{url}
\usepackage{balance}

\newtheorem{theorem}{Theorem}
\newtheorem{lemma}{Lemma}
\newtheorem{corollary}{Corollary}
\newtheorem{assumption}{Assumption}
\newtheorem{example}{Motivating Example}
\theoremstyle{definition}

\newcommand{\doop}{\mathit{do}}
\newcommand{\na}{\textcolor{gray}{--}}
\newcommand{\figsc}{0.91}

\def\BibTeX{{\rm B\kern-.05em{\sc i\kern-.025em b}\kern-.08em
    T\kern-.1667em\lower.7ex\hbox{E}\kern-.125emX}}

\makeatletter
\long\def\@makecaption#1#2{%
\ifx\@captype\@IEEEtablestring%
\footnotesize\bgroup\par\centering\@IEEEtabletopskipstrut{\normalfont\footnotesize #1}\\{\normalfont\footnotesize #2}\par\addvspace{0.5\baselineskip}\egroup%
\@IEEEtablecaptionsepspace
\else
\@IEEEfigurecaptionsepspace
\setbox\@tempboxa\hbox{\normalfont\footnotesize {#1.}\nobreakspace\nobreakspace #2}%
\ifdim \wd\@tempboxa >\hsize%
\setbox\@tempboxa\hbox{\normalfont\footnotesize {#1.}\nobreakspace\nobreakspace}%
\parbox[t]{\hsize}{\normalfont\footnotesize\noindent\unhbox\@tempboxa#2}%
\else%
\ifCLASSOPTIONconference \hbox to\hsize{\normalfont\footnotesize\hfil\box\@tempboxa\hfil}%
\else \hbox to\hsize{\normalfont\footnotesize\box\@tempboxa\hfil}%
\fi\fi\fi}
\makeatother

\begin{document}

\title{Identifying the Group to Intervene on to Maximize Effect Under Cross-Group Interference}

\author{
\IEEEauthorblockN{Xiaojing Du\textsuperscript{1}, Jiuyong Li\textsuperscript{1}, Lin Liu\textsuperscript{1}, Debo Cheng\textsuperscript{2}, Jixue Liu\textsuperscript{1}, Thuc Duy Le\textsuperscript{1}}
\IEEEauthorblockA{\small\textsuperscript{1}Adelaide University, Adelaide, Australia\\
\textsuperscript{2}Hainan University, Haikou, China\\
{\footnotesize \{xiaojing.du, jiuyong.li, lin.liu, jixue.liu, thuc.le\}@adelaide.edu.au, chengd@hainanu.edu.cn}}
\thanks{Correspondence to xiaojing.du@adelaide.edu.au.}
}


\maketitle

\begin{abstract}
In many networked systems, interventions applied to one group of units can induce substantial causal effects on another group through cross-group interference pathways. Despite its practical importance in domains such as public health, digital marketing, and social policy, the problem of identifying which intervention subset in a source group maximizes the benefit on a target group remains largely unaddressed. We formalize this problem as \emph{cross-group causal influence estimation} and introduce the \emph{core-to-group causal effect} (Co2G), a formally defined causal estimand that quantifies the contrast in target-group outcomes under intervention versus non-intervention on a candidate source subset. We establish the nonparametric identifiability of Co2G from observational network data using do-calculus under standard causal assumptions, and develop a graph neural network-based estimator that captures cross-group interference patterns. To navigate the combinatorial search space of candidate subsets, we propose \emph{CauMax}, an uncertainty-aware causal effect maximization framework with two scalable selection algorithms: (i)~\emph{CauMax-G}, an iterative greedy search with Monte Carlo dropout-based lower confidence bounds, and (ii)~\emph{CauMax-D}, a differentiable gradient-based optimization via a concrete Gumbel-Softmax relaxation. Extensive experiments on semi-synthetic data over two real-world social-network topologies, complemented by a real-outcome validation on three additional networks, show that CauMax substantially reduces regret relative to structural and diffusion-based heuristics, by up to an order of magnitude at small budgets and consistently across all budgets, while remaining robust to the uncertainty-penalty weight.
\end{abstract}


\section{Introduction}
In many real-world systems, the effect of an intervention does not stay with the units it targets. When a subset of units is treated, the effect travels along network ties and reshapes the outcomes of units in other groups, a pattern documented in social~\cite{ref_yang_neighbor}, economic~\cite{ref_han}, biological~\cite{ref_ma_mrsa}, and online~\cite{ref_yuan_yuan} systems. Treating one subset can thus shift the outcomes of units that are never directly treated, which gives rise to interference and violates the standard assumption that a unit's outcome depends only on its own treatment assignment. Understanding and quantifying causal effects under interference has therefore become a central challenge in modern causal inference on networked data~\cite{ref_imai,ref_imbens_book,ref_ogburn}.

\begin{figure}[t]
\centerline{\includegraphics[width=\figsc\columnwidth]{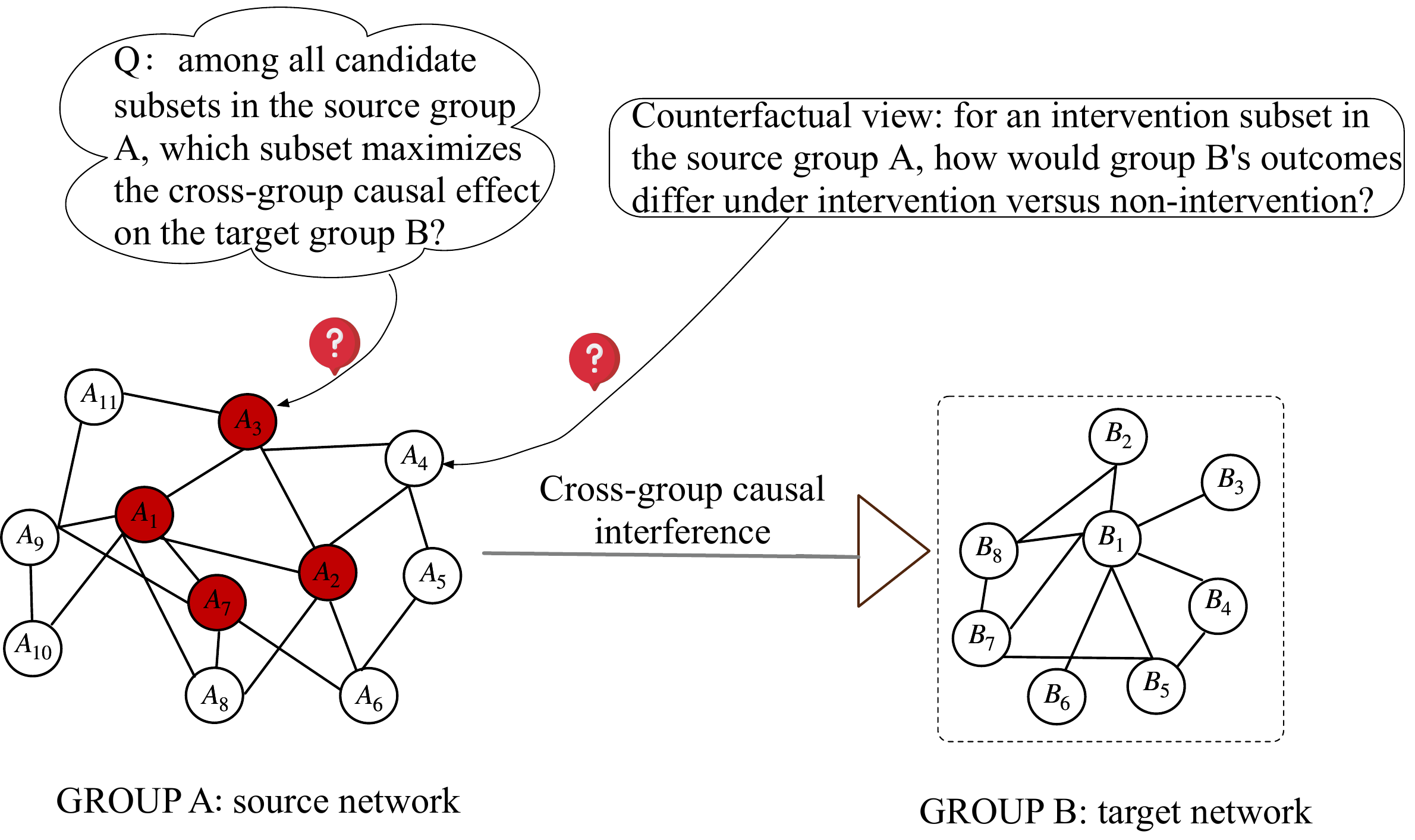}}
\caption{The cross-group causal interference problem. An intervention on a red candidate subset of source group $A$ affects target group $B$ through cross-group pathways, and we seek the subset maximizing the effect on $B$.}
\label{fig:overview}
\end{figure}

One prominent domain is epidemiology, where the role of \emph{core transmitters} in spreading infectious disease far exceeds their proportion in the population~\cite{ref_anderson_may,ref_boily}. A small, highly connected core drives most transmission, so interventions targeting it induce spillovers that propagate widely through contact networks~\cite{ref_cui_zhu,ref_schieber}. Under network interference, how can we formally identify which core group generates the largest causal influence, particularly when its own group produces cross-group spillovers? Such cross-group influence matters for targeting interventions wherever social, behavioral, or biological interactions span community boundaries.

\begin{example}[Public Vaccination]
The elderly or chronically ill often cannot be vaccinated directly, so their protection depends on intervening on core transmitters such as highly connected working-age adults, as illustrated in Fig.~\ref{fig:overview}. Different core subgroups yield markedly different target benefits, yet observed outcomes reflect only the implemented intervention, so policy requires the counterfactual of how they would differ under alternative choices.
\end{example}

\begin{example}[Digital Product Ecosystems]
A company may promote an app to students by providing a free version to selected content creators whose endorsement raises adoption among students. Creators differ widely in their impact on student adoption, yet the company observes outcomes for only one chosen set, so making an optimal decision requires the counterfactual of how adoption would differ under the endorsement of a different subset of creators.
\end{example}

Both examples share the same structure. An intervention on a source group propagates to a distinct target group, and a sound decision requires comparing counterfactual target-group outcomes across source subsets that are never jointly observed. Existing causal inference frameworks under network interference do not directly support this comparison. As summarized in Table~\ref{tab:related}, most prior studies either confine spillovers within a single community~\cite{ref_aronow,ref_cai_gen,ref_jiang_sun,ref_ogburn} or fail to capture the asymmetric influence of highly active core groups~\cite{ref_heitzig,ref_ji,ref_miguel,ref_possnig}. Recent work begins to formalize influential individuals' causal influence: Lee et al.~\cite{ref_lee} identify influential subjects within a single network, but restrict interference to within rather than across groups, and Su et al.~\cite{ref_su} cast influence-maximization gains as causal individual treatment effects under network interference, but model interference within a single community and do not capture asymmetric cross-group influence. Consequently, existing approaches cannot identify which core sources generate the largest cross-group causal influence, a capability needed for targeting interventions where communities interact asymmetrically.

Classical greedy influence-maximization (IM) methods select seeds to maximize expected cascade size under a prespecified diffusion model~\cite{ref_chen,ref_kempe}, which differs from our objective. Even perfectly predicting the observed diffusion cannot reveal how group $B$ would have responded had an alternative subset in group $A$ been selected, since such responses are potential outcomes never observed in real data~\cite{ref_imbens_book,ref_rubin}. Identifying the core source that maximizes cross-group improvement therefore requires a causal formulation, which defines the estimand of interest, clarifies what is identifiable from observational data, and enables principled comparison across interventions that were not carried out.

Identifying the source subset that maximizes cross-group improvement poses two challenges. \textbf{(a)~Unbiased estimation under cross-group interference:} treatments spill over through indirect pathways, so the potential outcomes of alternative subsets are never observed simultaneously, and standard estimators conflate causal spillovers with structure- or confounding-induced correlations. \textbf{(b)~Combinatorial search:} the number of candidate subsets grows exponentially in $|A|$, and classical IM heuristics optimize cascade sizes under assumed dynamics rather than the causal estimand. Together these motivate a framework that (i)~identifies unbiased cross-group causal effects from observational data and (ii)~efficiently finds the subset maximizing target improvement.

Motivated by these observations, we formulate the problem of cross-group causal influence maximization, which seeks to identify the core sources whose selection for intervention yields the greatest causal improvements in the outcomes of other groups. Our main contributions are summarized as follows:
\begin{itemize}
\item \textbf{Problem formulation.} We formalize cross-group causal influence maximization and define the core-to-group causal effect (Co2G), a causal estimand for the target-group improvement attributable to intervening on a source subset under interference.
\item \textbf{Nonparametric identifiability.} We establish identifiability of Co2G from observational network data via do-calculus under standard graphical assumptions, without parametric restrictions on the outcome model.
\item \textbf{Estimation and optimization.} We propose \emph{CauMax}, coupling a graph neural network (GNN) estimator of cross-group effects with an uncertainty-aware Monte Carlo dropout objective, and instantiate two scalable selectors: the greedy \emph{CauMax-G} and the differentiable \emph{CauMax-D}.
\item \textbf{Empirical evaluation.} On two real-world networks, with an additional real-outcome validation, CauMax cuts selection regret well below structural and diffusion-based baselines, with the largest margins at small budgets, and stays robust across the uncertainty-penalty weight.
\end{itemize}

\section{Related Work}
Our work lies at the intersection of causal inference under network interference and influence maximization on social networks. We review each line of work below.

\begin{table}[t]
\caption{Representative works on network interference and influence.}
\label{tab:related}
\centering
\setlength{\tabcolsep}{6pt}
\renewcommand{\arraystretch}{1.1}
\resizebox{\columnwidth}{!}{%
\begin{tabular}{lcc}
\toprule
\textbf{Representative works} & \textbf{Cross-group} & \textbf{Causal maximization}\\
\midrule
Jiang and Sun~\cite{ref_jiang_sun} & \ding{55} & \ding{55} \\
Cai et al.~\cite{ref_cai_gen} & \ding{55} & \ding{55} \\
Aronow and Samii~\cite{ref_aronow} & \ding{55} & \ding{55} \\
Ogburn et al.~\cite{ref_ogburn} & \ding{55} & \ding{55} \\
Miguel and Kremer~\cite{ref_miguel} & \checkmark & \ding{55} \\
Heitzig \& O'Keeffe~\cite{ref_heitzig} & \checkmark & \ding{55} \\
Ji et al.~\cite{ref_ji} & \checkmark & \ding{55} \\
Possnig et al.~\cite{ref_possnig} & \checkmark & \ding{55} \\
Lee et al.~\cite{ref_lee} & \ding{55} & \ding{55} \\
Su et al.~\cite{ref_su} & \ding{55} & \checkmark \\
\midrule
Ours & \checkmark & \checkmark \\
\bottomrule
\end{tabular}}
\end{table}

\subsection{Cross-Group Spillovers in Networked Systems}
Causal inference under interference began with the \emph{partial interference} assumption~\cite{ref_hudgens,ref_sobel,ref_imai}, where spillovers occur only within disjoint, pre-defined clusters. Later work relaxed this to general interference via exposure mappings~\cite{ref_aronow}, to neighbor and network interference on observational data~\cite{ref_forastiere,ref_ogburn}, and to representation-learning estimators on networks~\cite{ref_jiang_sun,ref_cai_gen}. Recent variants address detecting interference~\cite{ref_han}, fairness under networked interference~\cite{ref_yang_neighbor}, heterogeneous spillovers in A/B tests~\cite{ref_yuan_yuan}, intervention policy improvement~\cite{ref_ma_hyper}, separating peer direct from indirect effects~\cite{ref_du_peer}, peer effect estimation under simultaneous feedback and unobserved confounding~\cite{ref_du_feedback}, and treatment-effect estimation under dynamic network interference~\cite{huang2026modeling,lin2026dynamic}. These approaches model within-group or neighborhood-level effects and do not isolate a directional source-to-target cross-group causal channel.

Other work studies spillovers across communities empirically: cross-school deworming externalities~\cite{ref_miguel}, cross-village diffusion of savings groups~\cite{ref_heitzig}, and micro/macro-influencer spillovers~\cite{ref_ji}. On the estimation side, representation-learning estimators and generalization bounds for networked causal effects~\cite{ref_jiang_sun,ref_cai_gen} do not distinguish source and target groups or model asymmetric cross-group interference. In contrast, we formalize a directional source-to-target causal structure and introduce a combinatorial optimization objective over causal estimands.

\subsection{Influence Maximization}
Influence maximization selects seeds to maximize expected cascade spread under assumed diffusion models such as Independent Cascade (IC) and Linear Threshold~\cite{ref_chen,ref_kempe,ref_leskovec}. Recent work connects IM with causal reasoning, for example Lee et al.~\cite{ref_lee} for influential subjects within a single network and Su et al.~\cite{ref_su} for heterogeneous node-level gains. IM nonetheless differs from our objective. It optimizes expected activation counts under assumed dynamics, whereas we optimize a formally defined causal estimand, Co2G, over unobserved potential outcomes~\cite{ref_imbens_book,ref_rubin}. Our formulation thus combines causal identification with subset optimization over estimated effects, which diffusion-based influence maximization does not provide.

\section{Problem Setup and Cross-Group Interference}
This section formalizes the cross-group setting: the notation and asymmetric interference structure, the core-to-group causal effect (Co2G) and the combinatorial objective that maximizes it, and the assumptions under which Co2G is identifiable from observational network data.

\textbf{\textit{Notation.}} Let $V_A$ denote the set of source-group nodes and $V_B$ denote the set of target-group nodes. A candidate intervention subset is denoted by $S \subseteq V_A$ with a budget constraint $|S| \le K$. We use $\mathcal{S}$ to denote a collection of evaluated subsets, such as randomly sampled subsets of varying sizes, in empirical evaluation. For a fixed budget $k$, let $S_k$ denote the subset returned by an algorithm, and let $S_k^\star$ denote the oracle optimal subset under the ground-truth causal effect.

\textbf{\textit{Networked Groups and Interference Structure.}} We consider two disjoint groups, a \emph{source group} $V_A$ and a \emph{target group} $V_B$, with edges within and across groups, and write $G=(V_A\cup V_B,\,E)$ for the observed network topology. The interference is asymmetric. Interventions on $A$ may affect outcomes in $B$ but not vice versa, as arises in vaccination and disease control~\cite{ref_perez,ref_miguel}, influencer marketing~\cite{ref_ji}, and core-periphery systems~\cite{ref_kenett}.

\textbf{\textit{Subset-Level Intervention.}} We focus on deterministic joint interventions applied to a subset $S \subseteq V_A$. All units in $S$ are simultaneously assigned to the same treatment level, while units in $V_A \setminus S$ are fixed at a baseline control level. Throughout the paper, we assume a binary treatment with levels $t \in \{0,1\}$ and fix the baseline level as $t_0 = 0$. This is for clarity of exposition and not a restriction of the framework. The estimand, the identification result, and the estimator below all extend directly to a continuous treatment $t$. Our objective is to identify, among a combinatorial space of candidate subsets $S \subseteq V_A$, the subset-level intervention that maximizes the causal effect on the target group $B$.

\textbf{\textit{Interventional Target-Group Mean.}} Let $Y_j$ denote the outcome of target node $j \in V_B$, and define the average target-group outcome as $Y_B = \frac{1}{|V_B|}\sum_{j \in V_B} Y_j$. For any subset $S \subseteq V_A$ and treatment level $t \in \{0,1\}$, we define the \emph{interventional target-group mean} as
\begin{equation}
\mu_B(S; t) := \mathbb{E}\!\left[Y_B \,\middle|\, \doop(T_S=t), \doop(T_{V_A\setminus S}=t_0)\right],
\label{eq:mu_B}
\end{equation}

\textbf{\textit{Core-to-Group Causal Effect.}} The \emph{core-to-group causal effect} (Co2G) of subset $S$ at treatment level $t$ is the contrast between intervening on $S$ at level $t$ and at the baseline level $t_0$:
\begin{equation}
\mathrm{Co2G}(S;t) := \mu_B(S; t) - \mu_B(S; t_0),
\label{eq:co2g}
\end{equation}
the target-group gain from intervening on $S$ at level $t$. For a fixed active level $t$, we write $\mathrm{Co2G}(S)$ for $\mathrm{Co2G}(S;t)$.

\textbf{\textit{Cross-Group Causal Influence Maximization.}} Given an intervention budget $K$, the cross-group causal influence maximization problem seeks to identify
\begin{equation}
S^\star = \arg\max_{S \subseteq V_A,\, |S| \le K} \mathrm{Co2G}(S),
\label{eq:obj}
\end{equation}
that is, to select at most $K$ source nodes whose joint intervention yields the largest causal improvement in the outcomes of the target group $B$.

Let $\mathcal{G}$ be the causal directed acyclic graph (DAG) of Fig.~\ref{fig:dag}. We adopt three assumptions that are standard in causal inference.

\begin{assumption}[Causal Markov Property~\cite{ref_pearl}]\label{as:markov}
The joint distribution factorizes according to $\mathcal{G}$, so each variable is independent of its non-descendants given its parents. This standard condition defines a causal model and licenses the do-calculus used below.
\end{assumption}

\begin{assumption}[Network-conditional unconfoundedness~\cite{ref_guo,ref_veitch}]\label{as:unconf}
Any latent confounder such as homophily that jointly affects treatment and outcomes is mediated by node covariates and network position, so no unobserved common cause is left unadjusted.
\end{assumption}
This condition is standard for causal inference on networks. Our estimator conditions on both $X$ and $G$, and Section~\ref{sec:confound} quantifies the bias when it is violated.

\begin{assumption}[Positivity~\cite{ref_pearl}]\label{as:pos}
For any subset $S \subseteq V_A$ with $|S| \le K$, any treatment level $t \in \{0,1\}$, and any covariate realization $x$ with $p(X=x) > 0$, we assume $\Pr(T_S=t, T_{V_A\setminus S}=t_0 \mid X=x) > 0$.
\end{assumption}

\begin{figure*}[t]
\centerline{\includegraphics[width=\figsc\textwidth]{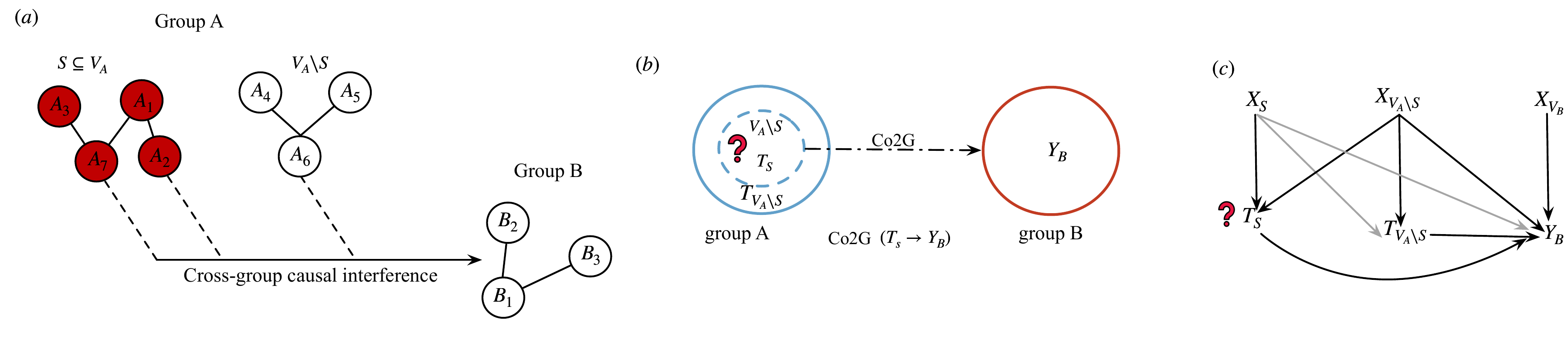}}
\caption{Cross-group causal structure underlying Co2G. \textbf{(a)}~The observed network conditioned on a candidate source subset $S\subseteq V_A$, shown in red: an intervention on $S$ propagates along the cross-group edges $E_{AB}$ to the target group $V_B$, while the complement $V_A\setminus S$ is held at the baseline level $t_0$. \textbf{(b)}~The core-to-group causal effect contrasts the average target-group outcome $Y_B$ under intervention on $S$ against non-intervention. Because the two potential outcomes are never observed jointly, this gain is counterfactual and must be inferred rather than read off the data. \textbf{(c)}~The causal DAG over the variables: pre-treatment covariates $X=(X_{V_A},X_{V_B})$ for source and target nodes, the source treatments $T_S$ and $T_{V_A\setminus S}$, and the target outcomes $Y_B$.}
\label{fig:dag}
\end{figure*}

\section{Method}
In this section, we present the full framework for cross-group causal effect estimation and optimization. We first establish the nonparametric identifiability of Co2G from observational data (Section~\ref{sec:ident}), then develop a learning-based estimator and optimization procedure consisting of three components: (i)~a graph neural network estimator that captures cross-group interference patterns (Section~\ref{sec:estimator}), (ii)~an uncertainty-aware optimization objective that guards against out-of-distribution predictions (Section~\ref{sec:objective}), and (iii)~two scalable search algorithms for subset selection (Section~\ref{sec:algorithms}).

\subsection{Identifiability of the Cross-Group Causal Effect}
\label{sec:ident}
We establish the identifiability of the interventional target-group mean $\mu_B(S;t)$ and the corresponding core-to-group causal effect $\mathrm{Co2G}(S)$ defined in Eqs.~\eqref{eq:mu_B}--\eqref{eq:co2g} under the cross-group causal structure illustrated in Fig.~\ref{fig:dag}. Let $X = (X_{V_A}, X_{V_B})$ denote the collection of observed pre-treatment covariates, where $X_{V_A}$ denotes covariates of source-group nodes and $X_{V_B}$ denotes covariates of target-group nodes. By construction, all variables in $X$ are measured prior to treatment assignment and are not affected by any intervention.

Our goal is to identify, from observational data, the interventional mean $\mu_B(S;t) = \mathbb{E}[Y_B \mid \doop(T_S=t), \doop(T_{V_A\setminus S}=t_0)]$, for any subset $S \subseteq V_A$ with $|S| \le K$ and treatment level $t \in \{0,1\}$.

\begin{theorem}[Identifiability of the Interventional Target-Group Mean]
\label{thm:ident}
Under Assumptions~\ref{as:markov}--\ref{as:pos} and the causal graph in Fig.~\ref{fig:dag}, with the network $G$ fixed and observed, the interventional target-group mean $\mu_B(S; t)$ is identifiable from the observational distribution $p(X, T_{V_A}, Y_B)$ and admits the following representation:
\begin{equation}
\mu_B(S; t) = \sum_{x} \mathbb{E}\!\left[Y_B \mid T_S{=}t,\, T_{V_A\setminus S}{=}t_0,\, X{=}x\right] p(X{=}x).
\label{eq:ident_mu}
\end{equation}
\end{theorem}

\begin{proof}
Write $\bar S := V_A\setminus S$ and let $\doop(t)$ abbreviate the joint intervention $\doop(T_S{=}t),\doop(T_{\bar S}{=}t_0)$. Applying the law of total probability over the pre-treatment covariates $X$ gives $P(y_B \mid \doop(t)) = \sum_x P(y_B \mid X{=}x, \doop(t))\, P(X{=}x \mid \doop(t))$, which we identify factor by factor.

For the covariate factor, Rule 3 of do-calculus, with $\overline{T}$ denoting the removal of incoming edges into $T$, gives $P(X{=}x \mid \doop(t))=P(X{=}x)$, since $X \perp\!\!\!\perp (T_S, T_{\bar S})$ in $\mathcal{G}_{\overline{T_S}, \overline{T_{\bar S}}}$.

For the outcome factor, we exchange interventions for observations by Rule 2, with $\underline{T}$ denoting the removal of outgoing edges from $T$: $P(y_B \mid X{=}x, \doop(t)) = P(y_B \mid X{=}x, T_S{=}t, T_{\bar S}{=}t_0)$, since $Y_B \perp\!\!\!\perp T_{\bar S} \mid (T_S, X)$ in $\mathcal{G}_{\overline{T_S},\,\underline{T_{\bar S}}}$ and $Y_B \perp\!\!\!\perp T_S \mid (T_{\bar S}, X)$ in $\mathcal{G}_{\underline{T_S}}$. Combining the two factors yields Eq.~\eqref{eq:ident_mu}.
\end{proof}

\begin{corollary}[Identifiability of the Core-to-Group Causal Effect]
\label{cor:ident_co2g}
For any subset $S \subseteq V_A$ with $|S| \le K$, writing $\bar S = V_A\setminus S$, the core-to-group causal effect $\mathrm{Co2G}(S;t)=\mu_B(S;t)-\mu_B(S;t_0)$ is identifiable and can be expressed as
\begin{align}
\mathrm{Co2G}(S;t) = \sum_{x} p(x)\big(
& \mathbb{E}[Y_B \mid T_S{=}t, T_{\bar S}{=}t_0, X{=}x] \nonumber\\[-2pt]
&{}- \mathbb{E}[Y_B \mid T_S{=}t_0, T_{\bar S}{=}t_0, X{=}x]\big).
\label{eq:co2g_ident}
\end{align}
\end{corollary}

\subsection{Causal Effect Estimator}
\label{sec:estimator}
As established in Section~\ref{sec:ident}, the interventional target-group mean $\mu_B(S;t)$ and the core-to-group causal effect $\mathrm{Co2G}(S)$ are identifiable from observational data under Assumptions~\ref{as:markov}--\ref{as:pos}. We now introduce a parameterized estimator designed to approximate these identified causal quantities from networked observational data.

\textbf{\textit{Estimator Overview.}} Let $f_\theta$ denote a learned causal effect estimator. Given a subset $S \subseteq V_A$ and a treatment level $t \in \{0,1\}$, we represent the corresponding joint intervention $\doop(T_S=t, T_{V_A\setminus S}=t_0)$ by a binary treatment vector $T(S,t) \in \{0,1\}^{|V_A|}$, where
\begin{equation}
T_i(S,t)=\begin{cases} t, & i \in S, \\ t_0, & i \in V_A \setminus S, \end{cases}\qquad t_0=0.
\end{equation}
The estimator $f_\theta$ takes as input the network structure, node covariates, and the subset-induced treatment vector $T(S,t)$, and outputs a prediction of the corresponding interventional target-group mean:
\begin{equation}
\hat{\mu}_B(S; t) = f_\theta(T(S,t)).
\end{equation}
The estimated core-to-group causal effect is then given by
\begin{equation}
\widehat{\mathrm{Co2G}}(S) = \hat{\mu}_B(S;t) - \hat{\mu}_B(S;t_0).
\end{equation}

\textbf{\textit{Model Architecture.}} The estimator $f_\theta$ is a structured realization of the identified formula in Eq.~\eqref{eq:ident_mu}. It parameterizes the conditional mean $\mathbb{E}[Y_B\mid T_S, T_{V_A\setminus S}, X]$ following Fig.~\ref{fig:dag}(c), so that setting $T(S,t)$ and differencing $\hat{\mu}_B(S;t)-\hat{\mu}_B(S;t_0)$ recovers $\widehat{\mathrm{Co2G}}(S)$. Reading the graph from source to target gives three stages. Stage (i) is a source encoder that maps $X_{V_A}$ to per-source quantities. Stage (ii) is a cross-group aggregation that carries them along the cross-edges $E_{AB}$ to each target, where the treatment enters. Stage (iii) is a target predictor that maps the resulting exposure and $X_{V_B}$ to the outcome.

\textbf{\textit{(i)~Source encoder.}} A GNN~\cite{ref_wu_gnn} over the within-source graph maps $X_{V_A}$ to embeddings $g_i$. Source and non-source covariates are encoded together, with the $S$ vs.\ $V_A\setminus S$ split carried by $T(S,t)$ in stage~(ii), and the treatment is not an encoder input. From $h_i = [X_i;\, g_i]$, a skip connection preserving the pointwise feature signal, we read a nonnegative per-source effectiveness $\hat{a}_i = \phi(w_a^\top h_i)$ with $\phi$ the softplus, and a feature-spillover term $b_i = W_b h_i \in \mathbb{R}^{d_b}$. The learned $\hat{a}_i$, estimated from data unlike the ground-truth $a_i$ of Section~\ref{sec:setup}, allows the model to represent influence that is heterogeneous across sources and need not coincide with high degree.

\textbf{\textit{(ii)~Cross-group interaction.}} For each target $j \in V_B$ with source neighbors $N_A(j) = \{i \in V_A : (i,j) \in E_{AB}\}$ and weighted cross-degree $D_j = \sum_{i \in N_A(j)} w_{ij}$, we form the scalar treatment exposure $e_j(T) = \frac{1}{\sqrt{D_j}}\sum_{i \in N_A(j)} \hat{a}_i\, w_{ij}\, T_i$ and a treatment-independent feature-spillover summary $q_j = \frac{1}{\sqrt{D_j}}\sum_{i \in N_A(j)} w_{ij}\, b_i$, matching the $\sqrt{D_j}$ normalization of the interference mechanism.

\textbf{\textit{(iii)~Target predictor.}} The target covariates enter directly here, with no target-side message passing. Each target's exposures are concatenated with its raw covariates $X_j$ and passed through a multilayer perceptron (MLP) head, $\widehat{Y}_j(S;t) = \mathrm{MLP}\big([X_j,\, e_j(T(S,t)),\, q_j]\big)$, and the interventional target-group mean is
\begin{equation}
\hat{\mu}_B(S;t) = \frac{1}{|V_B|} \sum_{j \in V_B} \widehat{Y}_j(S;t).
\end{equation}
Stages~(i)--(iii) together model the identified conditional mean of Eq.~\eqref{eq:ident_mu}. Here $\hat{\mu}_B(S;t)$ estimates $\mu_B(S;t)$, and the difference $\hat{\mu}_B(S;t)-\hat{\mu}_B(S;t_0)$ estimates the identified contrast $\mathrm{Co2G}(S)$ of Eq.~\eqref{eq:co2g_ident}. The estimator is trained on observational data by regressing the predicted outcomes onto the observed target outcomes $Y_{V_B}^{\mathrm{obs}}$ under the observed assignment $T^{\mathrm{obs}}$, using a mean-squared-error loss with early stopping on held-out targets.

The following lemma bounds the resulting estimation error by the estimator's outcome-prediction error, which obeys a standard generalization rate.

\begin{lemma}[Co2G estimation error]
\label{lem:bound}
Let $\mathcal{F}$ denote the hypothesis class of the outcome predictor $f_\theta$. Suppose the counterfactual exposures induced by the treated assignment $T(S,t)$ and the baseline assignment $T(S,t_0)$ lie within the support of the training-exposure distribution, a condition we call exposure overlap, and that $f_\theta$ predicts the per-target outcomes with root-mean-squared error $\varepsilon_n$ over the $n=|V_B|$ target nodes, uniformly over both assignments. Then for any subset $S$,
\begin{equation}
\big|\widehat{\mathrm{Co2G}}(S) - \mathrm{Co2G}(S)\big| \le 2\,\varepsilon_n,
\end{equation}
and, by uniform convergence over $\mathcal{F}$, with probability at least $1-\delta$,
\begin{equation}
\varepsilon_n \le \widehat{\varepsilon}_n + 2\,\widehat{\mathfrak{R}}_n(\mathcal{F}) + O\!\big(\sqrt{\log(1/\delta)/n}\big),
\end{equation}
where $\widehat{\varepsilon}_n$ is the empirical training error and $\widehat{\mathfrak{R}}_n(\mathcal{F})$ is the empirical Rademacher complexity of the hypothesis class $\mathcal{F}$.
\end{lemma}
\begin{proof}
The first inequality holds because both $\mathrm{Co2G}$ and $\widehat{\mathrm{Co2G}}$ are averages over $V_B$ of per-target outcome contrasts, so their difference is at most twice the per-target prediction error. The second is the standard Rademacher-complexity generalization bound~\cite{ref_bartlett_rademacher} applied to the predictor class $\mathcal{F}$.
\end{proof}
Hence the generalization term of the Co2G estimation error scales as $O(n^{-1/2})$ up to the representation complexity, and Section~\ref{sec:converge} reports a monotone decrease of the held-out error with $n$ that is consistent with this upper bound. In practice, subsets whose induced exposures lie far from the observed support tend to receive a larger $\hat{\sigma}_{\mathrm{Co2G}}$ and are penalized by the lower-confidence-bound objective of Section~\ref{sec:objective}.

\subsection{Uncertainty-Aware Optimization Objective}
\label{sec:objective}
Given the estimator $\widehat{\mathrm{Co2G}}(S)$, a natural approach would be to select the subset $S$ that maximizes the predicted causal effect. However, since $f_\theta$ is learned from finite observational data, predictions for subsets that are far from the support of the observed data may suffer from high epistemic uncertainty.

To account for this uncertainty, we adopt a risk-averse optimization criterion based on a lower confidence bound (LCB)~\cite{ref_jin_pessimism}. For a candidate subset $S \subseteq V_A$, we perform $M$ stochastic forward passes using Monte Carlo (MC) dropout~\cite{ref_gal} and obtain samples $\{\widehat{\mathrm{Co2G}}^{(m)}(S)\}_{m=1}^{M}$. We compute the aggregated empirical mean and standard deviation as
\begin{align}
\hat{\mu}_{\mathrm{Co2G}}(S) &= \frac{1}{M}\sum_{m=1}^{M} \widehat{\mathrm{Co2G}}^{(m)}(S), \\
\hat{\sigma}_{\mathrm{Co2G}}(S) &= \sqrt{\tfrac{1}{M}\sum_{m=1}^{M}\big(\widehat{\mathrm{Co2G}}^{(m)}(S) - \hat{\mu}_{\mathrm{Co2G}}(S)\big)^{2}}.
\end{align}
The uncertainty-aware objective is defined as
\begin{equation}
J(S) = \hat{\mu}_{\mathrm{Co2G}}(S) - \lambda\,\hat{\sigma}_{\mathrm{Co2G}}(S),
\label{eq:objective}
\end{equation}
where $\lambda \ge 0$ controls the strength of uncertainty penalization. Maximizing $J(S)$ favors subsets with large predicted causal effects while discouraging selections associated with high uncertainty.

\subsection{Subset Selection Algorithms}
\label{sec:algorithms}
\textbf{\textit{Greedy Search.}} For small to medium-sized networks, we adopt an iterative greedy procedure to approximately solve $\max_{S \subseteq V_A} J(S)$ subject to $|S|\le K$. Starting from the empty set $S_0=\emptyset$, at iteration $k$ we select the node that yields the largest marginal improvement in the uncertainty-aware objective:
\begin{equation}
v^\star = \arg\max_{v\in V_A \setminus S_{k-1}} \Big[J(S_{k-1} \cup \{v\}) - J(S_{k-1})\Big].
\label{eq:greedy_step}
\end{equation}
If the maximum marginal gain is positive, we update $S_k = S_{k-1}\cup\{v^\star\}$. Otherwise the search stops before reaching $K$. At each step we evaluate the marginal gain of every source with nonzero cross-group reach and impose no degree-based candidate restriction. Because effectiveness is feature-driven, pruning candidates by cross-group degree can exclude exactly the effective but low-reach sources the estimand targets and leave a non-negligible optimality gap at larger budgets, so we search the full source set. For evaluation at a fixed budget $K$, we then pad any early-stopped selection to size $K$ with the highest cross-group-degree unselected sources, so that all budget levels are directly comparable. The full procedure is summarized in Algorithm~\ref{alg:greedy}.

\textbf{\textit{Differentiable Subset Optimization.}} For large-scale source groups, the greedy procedure in Algorithm~\ref{alg:greedy} becomes computationally expensive due to repeated evaluations of $J(S)$. To address this limitation, we introduce a differentiable relaxation that enables gradient-based optimization of the same uncertainty-aware objective.

Instead of directly optimizing over discrete subsets, we relax the subset indicator to a continuous vector and optimize a smooth surrogate of $J(S)$ while enforcing the budget constraint. After optimization, a discrete subset is recovered via a top-$K$ projection. The resulting procedure is summarized in Algorithm~\ref{alg:diff}. Each source node $i \in V_A$ is associated with a learnable logit $\psi_i$, and a relaxed treatment vector $\tilde{T} \in [0,1]^{|V_A|}$ is obtained through a concrete binary Gumbel-Softmax relaxation~\cite{ref_jang}, $\tilde{T}_i = \sigma(\psi_i/\tau)$ with $\sigma$ the logistic sigmoid, which is fed to the estimator as a continuous treatment in the exposure $e_j(\tilde{T})$. We use the deterministic relaxation with temperature annealing, and the estimator $f_\theta$ is kept fixed while the logits are optimized under a soft budget constraint:
\begin{equation}
\min_{\psi} \;-J(\tilde{T}(\psi)) + \gamma\,\big(\|\tilde{T}(\psi)\|_1 - K\big)^2 .
\label{eq:diff_loss}
\end{equation}
The logits are initialized so that $\|\tilde{T}\|_1 \approx K$ at the start, which keeps the budget penalty from dominating the objective. After optimization, a discrete subset $S$ is obtained by selecting the top-$K$ nodes according to $\psi$. This approach scales linearly with $|V_A|$ and is suitable for large-scale networks.

\begin{algorithm}[t]
\caption{Uncertainty-aware greedy search (CauMax-G)}
\label{alg:greedy}
\footnotesize
\begin{algorithmic}[1]
\REQUIRE Trained estimator $f_\theta$; source nodes $V_A$; budget $K$; MC passes $M$; uncertainty penalty $\lambda$.
\ENSURE Selected subset $S_K$.
\STATE $S \leftarrow \emptyset$.
\STATE Estimate $J(S) = \hat{\mu}_{\mathrm{Co2G}}(S) - \lambda\,\hat{\sigma}_{\mathrm{Co2G}}(S)$ using $M$ MC-dropout passes.
\FOR{$k = 1$ to $K$}
  \STATE $\Delta^\star \leftarrow -\infty$, $v^\star \leftarrow \varnothing$.
  \FORALL{$v \in V_A \setminus S$}
    \STATE Estimate $J(S \cup \{v\})$ using $M$ MC-dropout passes.
    \STATE $\Delta \leftarrow J(S\cup\{v\}) - J(S)$.
    \IF{$\Delta > \Delta^\star$}
      \STATE $\Delta^\star \leftarrow \Delta$, $v^\star \leftarrow v$.
    \ENDIF
  \ENDFOR
  \IF{$\Delta^\star > 0$}
    \STATE $S \leftarrow S \cup \{v^\star\}$.
    \STATE Update cached $J(S)$.
  \ELSE
    \STATE \textbf{break}.
  \ENDIF
\ENDFOR
\STATE $S_K \leftarrow S$, padded to size $K$ with the highest cross-group-degree unselected sources if $|S|<K$.
\RETURN $S_K$.
\end{algorithmic}
\end{algorithm}

\begin{algorithm}[t]
\caption{Differentiable gradient search (CauMax-D)}
\label{alg:diff}
\footnotesize
\begin{algorithmic}[1]
\REQUIRE Effect estimator $f_\theta$; source nodes $V_A$; budget $K$; iterations $M_\mathrm{iter}$; learning rate $\eta$; MC passes $M$; uncertainty penalty $\lambda$; regularization $\gamma$; initial/final temperatures $\tau_0,\tau_T$.
\ENSURE Selected subset $S_K$.
\STATE Initialize $\psi_i \leftarrow \sigma^{-1}(K/|V_A|)$ for all $i$, so that $\|\tilde{T}\|_1 \approx K$.
\FOR{$\ell = 1$ to $M_\mathrm{iter}$}
  \STATE Anneal temperature $\tau \leftarrow \max(\tau_T,\, \tau_0 - \tfrac{\ell}{M_\mathrm{iter}}(\tau_0-\tau_T))$.
  \STATE Generate soft mask $\tilde{T} \leftarrow \sigma(\psi/\tau) \in [0,1]^{|V_A|}$.
  \STATE Estimate $\hat{\mu}_{\mathrm{Co2G}}(\tilde{T})$ and $\hat{\sigma}_{\mathrm{Co2G}}(\tilde{T})$ using $M$ MC-dropout passes.
  \STATE $J(\tilde{T}) \leftarrow \hat{\mu}_{\mathrm{Co2G}}(\tilde{T}) - \lambda\,\hat{\sigma}_{\mathrm{Co2G}}(\tilde{T})$.
  \STATE $\mathcal{L} \leftarrow -J(\tilde{T}) + \gamma(\|\tilde{T}\|_1 - K)^2$.
  \STATE $\psi \leftarrow \psi - \eta \nabla_\psi \mathcal{L}$.
\ENDFOR
\STATE $S_K \leftarrow \mathrm{TopK\_Indices}(\psi, K)$.
\RETURN $S_K$.
\end{algorithmic}
\end{algorithm}

\subsection{Computational Complexity}
\label{sec:complexity}
With $C_\mathrm{fwd}$ the cost of one forward pass and $M$ MC-dropout passes, CauMax-G costs $O(K\,P\,M\,C_\mathrm{fwd})$ with $P$ candidates evaluated per step, where $P{=}|V_A|$ in our experiments with no restriction. CauMax-D is $O(M_\mathrm{iter}\,M\,C_\mathrm{fwd})$, independent of $|V_A|$ and thus the most scalable estimator-based method. Structural baselines are cheaper, with Degree at $O(|V_A|\log|V_A|+|E|)$ and IM at $O(K\,|E_{AB}|)$, but ignore the causal effect.

\section{Experiments}
We evaluate the proposed framework on five real-world networks. Two social networks, BlogCatalog (BC) and Flickr~\cite{ref_guo}, drive the semi-synthetic regret study, where ground-truth Co2G is available by construction. Three additional networks with real node outcomes, the GitHub developer graph and the two Twitch streamer graphs Twitch-ES and Twitch-ENGB~\cite{ref_musae}, drive a complementary validation (Section~\ref{sec:real}). Our experiments address five questions. \textbf{Q1}: does the estimator accurately recover cross-group causal effects? \textbf{Q2}: do the causal selection methods reduce regret over structural and diffusion-based baselines? \textbf{Q3}: how do the two optimization strategies compare, and how sensitive are they to the uncertainty penalty? \textbf{Q4}: how robust is the framework to violations of unconfoundedness, that is, hidden confounding? \textbf{Q5}: is the estimated cross-group effect supported on networks with real outcomes? Our data, code, and the exact reproduction parameters, namely random seeds, hyperparameters, and run commands, are released at \url{https://github.com/dudaren/CauMax-ICDM2026}.

\subsection{Experimental Setup}
\label{sec:setup}
Both semi-synthetic benchmarks are built on real social-network topologies, BC with $5{,}186$ nodes and Flickr with $5{,}912$, split into a source group $V_A$, the top $p_A\%$ of nodes by $k$-core coreness with $p_A{=}15$, and target group $V_B=V\setminus V_A$. We keep the cross-group edges $E_{AB}$ and the target nodes with at least one source neighbor. We generate semi-synthetic observational data from a structural causal model (SCM) on this graph: source $i$ is treated with propensity $\pi_i=\sigma(w_\pi^\top X_i)$ with $\sigma$ the logistic function, each cross-edge has a heterogeneous weight $w_{ij}\sim\mathrm{Lognormal}(0,\sigma_w^2)$ with $\sigma_w{=}0.8$, and each source has a feature-dependent effectiveness $a_i=\phi(c\,\bar z_i)$ with $\bar z_i$ the standardized $w_z^\top X_i$, $\phi$ the softplus, and scale $c{=}3$. We fix these by realism rather than tuning: $\sigma_w{=}0.8$ spreads edge weights over an order of magnitude and $c{=}3$ makes influence feature- rather than degree-driven, and Section~\ref{sec:rho} sweeps the full degree-alignment spectrum, so the conclusion does not depend on this choice. With $D_j=\sum_{i\in N_A(j)} w_{ij}$, the target outcome is $y_j = w_y^\top X_j + \beta\, \phi(\tfrac{1}{\sqrt{D_j}}\sum_{i \in N_A(j)} a_i w_{ij} T_i) + \alpha\, \phi(\tfrac{1}{\sqrt{D_j}}\sum_{i \in N_A(j)} w_{ij} w_x^\top X_i) + \epsilon_j$, with $\beta{=}1$, $\alpha{=}0.5$, and $\epsilon_j\sim\mathcal{N}(0,0.1^2)$. The estimator is trained only on $(X,T^{\mathrm{obs}},Y_{V_B}^{\mathrm{obs}})$, and the ground-truth Co2G follows from the same equation. The estimator uses hidden width $128$ and $400$ training epochs with Adam, learning rate $10^{-3}$, and dropout $0.1$. Selection uses $M{=}20$ MC-dropout passes and $\lambda{=}0.5$, budgets are $K\in\{5,10,15,20,30,50\}$, and all results are mean$\,\pm\,$std over five independent SCM draws.

\textbf{\textit{Implementation and reproducibility.}}
The target nodes are split $80/20$ into training and validation sets, with early stopping on the validation outcome loss. There is no separate held-out outcome test split because the estimator is never trained on any causal quantity. The held-out evaluation is the causal task itself, namely the selection regret against the ground-truth oracle and the Co2G RMSE over $200$ independently sampled subsets, neither of which is seen during outcome-regression training.
We perform no per-dataset hyperparameter tuning. The uncertainty penalty $\lambda$ is swept over $\{0,0.25,0.5,1,2\}$ (Section~\ref{sec:lambda}), and all other settings are fixed to the defaults above, while CauMax-D additionally uses $500$ iterations, logit learning rate $0.05$, temperature annealed $1.0{\to}0.3$, and budget penalty $\gamma{=}10$. We run the full matrix of jobs across eight NVIDIA GTX 1080~Ti GPUs. Training an estimator takes seconds, while a budget selection takes under $3$ min for CauMax-D and ${\approx}1$--$2$\,h for the full $K{=}50$ CauMax-G search over the unrestricted source set, versus under $0.2$s for the structural baselines.

\subsection{Baselines}
We compare against five baselines. \textit{Random} selects $K$ source nodes uniformly at random~\cite{ref_kempe}. \textit{Degree}~\cite{ref_kempe} selects the $K$ source nodes with the highest cross-group out-degree. \textit{Weighted-Degree} selects the $K$ sources with the highest weighted cross-group out-degree $\sum_{j} w_{ij}$, a stronger structural heuristic under heterogeneous edges. \textit{Influence Maximization} selects nodes to maximize diffusion-based cascade spread without considering causal effects. We instantiate it as a weighted bipartite Independent Cascade model~\cite{ref_kempe} optimized with the Cost-Effective Lazy Forward (CELF) lazy-greedy algorithm~\cite{ref_leskovec}. The fifth, \textit{NetEst-G}, is a causal baseline in the spirit of NetEst~\cite{ref_jiang_sun}, a two-layer graph convolutional network (GCN) with a mean-pooled treated-neighbor exposure. We evaluate \emph{CauMax-G} (Algorithm~\ref{alg:greedy}) and the differentiable \emph{CauMax-D} (Algorithm~\ref{alg:diff}). Owing to space, the complete baseline configurations are provided in the public repository.

\subsection{Evaluation Metrics}
\label{sec:metrics}
\textbf{\textit{Approximate Oracle Subset ($S_k^\star$).}} Since exhaustive enumeration over the $\binom{|V_A|}{k}$ subsets is infeasible, we approximate the oracle $S_k^\star=\arg\max_{|S|=k}\mathrm{Co2G}(S)$ by \emph{Oracle-Greedy}, which greedily adds the source with the largest marginal increase in the true $\mathrm{Co2G}$. This is used only for evaluation.

\textbf{\textit{Regret@k.}} Given a method that returns a subset $S_k$ of size $k$, we define the regret as $\mathrm{Regret}@k = \mathrm{Co2G}(S_k^\star) - \mathrm{Co2G}(S_k)$, which measures the performance gap between the selected subset and the approximate oracle. Lower regret indicates that the selected subset achieves a causal effect closer to the oracle solution.

\textbf{\textit{Estimation Error.}} In addition to subset-level performance, we evaluate the accuracy of the causal effect estimator by computing the root mean squared error (RMSE) between the estimated and true core-to-group causal effects over a collection of evaluated subsets $\mathcal{S}$:
\begin{equation}
\mathrm{RMSE} = \sqrt{\frac{1}{|\mathcal{S}|}\sum_{S\in\mathcal{S}}\big(\widehat{\mathrm{Co2G}}(S) - \mathrm{Co2G}(S)\big)^2}.
\end{equation}

\subsection{Main Results}
\label{sec:results}
\textbf{\textit{Estimation accuracy (Q1).}} Table~\ref{tab:rmse} reports the Co2G estimation RMSE. It is small relative to the oracle effect and grows gently with the budget. The structured CauMax estimator recovers Co2G more accurately than the mean-pooled NetEst-G at nearly every budget.

\textbf{\textit{Selection regret (Q2).}} Table~\ref{tab:main} reports Regret@$K$ on BC and Flickr at $\lambda{=}0.5$. Across both datasets and all budgets, the estimator-based methods achieve markedly lower regret than Weighted-Degree, Degree, IM, and Random. The gap is largest at small budgets. On BC at $K{=}5$, CauMax-G reaches $0.009$ against $0.15$--$0.18$ for the heuristics, over an order of magnitude lower, and it persists at every budget. CauMax-G and CauMax-D take the top two positions at every budget on both datasets, with CauMax-G best and CauMax-D consistently second.

\begin{table}[t]
\caption{Co2G estimation RMSE@$K$. CauMax-G and CauMax-D share one estimator.}
\label{tab:rmse}
\centering\scriptsize
\setlength{\tabcolsep}{3pt}
\renewcommand{\arraystretch}{0.9}
\begin{tabular}{ll cccccc}
\toprule
\textbf{Dataset} & \textbf{Method} & $K{=}5$ & $K{=}10$ & $K{=}15$ & $K{=}20$ & $K{=}30$ & $K{=}50$ \\
\midrule
\multirow{2}{*}{BC}
& CauMax-G/-D     & \textbf{0.009} & \textbf{0.014} & \textbf{0.018} & \textbf{0.022} & \textbf{0.029} & \textbf{0.043} \\
& NetEst-G        & 0.011 & 0.018 & 0.023 & 0.028 & 0.038 & 0.056 \\
\midrule
\multirow{2}{*}{Flickr}
& CauMax-G/-D     & \textbf{0.020} & \textbf{0.031} & \textbf{0.041} & \textbf{0.049} & \textbf{0.065} & 0.097 \\
& NetEst-G        & 0.024 & 0.034 & 0.042 & 0.051 & 0.066 & \textbf{0.094} \\
\bottomrule
\end{tabular}
\end{table}

\begin{table*}[t]
\caption{Regret@$K$ on BC and Flickr at $\lambda{=}0.5$. \textbf{Bold}/\underline{underline} = best/second-best per column, lower is better.}
\label{tab:main}
\centering
\scriptsize
\setlength{\tabcolsep}{2.5pt}
\renewcommand{\arraystretch}{0.95}
\begin{tabular}{ll cccccc}
\toprule
\textbf{Dataset} & \textbf{Method} & $K{=}5$ & $K{=}10$ & $K{=}15$ & $K{=}20$ & $K{=}30$ & $K{=}50$ \\
\midrule
\multirow{7}{*}{BC}
& CauMax-G        & \textbf{0.009}$\pm$0.008 & \textbf{0.034}$\pm$0.013 & \textbf{0.057}$\pm$0.020 & \textbf{0.071}$\pm$0.021 & \textbf{0.104}$\pm$0.024 & \textbf{0.160}$\pm$0.022 \\
& CauMax-D        & \underline{0.022}$\pm$0.011 & \underline{0.047}$\pm$0.024 & \underline{0.073}$\pm$0.027 & \underline{0.096}$\pm$0.027 & \underline{0.118}$\pm$0.028 & \underline{0.259}$\pm$0.021 \\
\cmidrule(l){2-8}
& NetEst-G        & 0.073$\pm$0.033 & 0.111$\pm$0.045 & 0.152$\pm$0.052 & 0.177$\pm$0.039 & 0.256$\pm$0.045 & 0.406$\pm$0.081 \\
& Weighted-Degree & 0.147$\pm$0.050 & 0.244$\pm$0.073 & 0.328$\pm$0.075 & 0.388$\pm$0.098 & 0.504$\pm$0.116 & 0.680$\pm$0.077 \\
& Degree          & 0.156$\pm$0.056 & 0.251$\pm$0.077 & 0.348$\pm$0.082 & 0.435$\pm$0.085 & 0.511$\pm$0.074 & 0.749$\pm$0.081 \\
& IM              & 0.149$\pm$0.053 & 0.233$\pm$0.069 & 0.298$\pm$0.081 & 0.351$\pm$0.075 & 0.464$\pm$0.074 & 0.645$\pm$0.043 \\
& Random          & 0.176$\pm$0.033 & 0.297$\pm$0.044 & 0.408$\pm$0.054 & 0.498$\pm$0.063 & 0.658$\pm$0.081 & 0.932$\pm$0.104 \\
\midrule
\multirow{7}{*}{Flickr}
& CauMax-G        & \textbf{0.103}$\pm$0.063 & \textbf{0.158}$\pm$0.118 & \textbf{0.220}$\pm$0.069 & \textbf{0.296}$\pm$0.059 & \textbf{0.433}$\pm$0.093 & \textbf{0.527}$\pm$0.119 \\
& CauMax-D        & \underline{0.109}$\pm$0.067 & \underline{0.227}$\pm$0.127 & \underline{0.285}$\pm$0.099 & \underline{0.386}$\pm$0.123 & \underline{0.515}$\pm$0.100 & \underline{0.606}$\pm$0.114 \\
\cmidrule(l){2-8}
& NetEst-G        & 0.394$\pm$0.144 & 0.592$\pm$0.198 & 0.618$\pm$0.154 & 0.705$\pm$0.155 & 0.860$\pm$0.213 & 0.984$\pm$0.256 \\
& Weighted-Degree & 0.366$\pm$0.131 & 0.494$\pm$0.151 & 0.645$\pm$0.157 & 0.665$\pm$0.185 & 0.741$\pm$0.199 & 0.818$\pm$0.134 \\
& Degree          & 0.368$\pm$0.131 & 0.505$\pm$0.145 & 0.645$\pm$0.182 & 0.655$\pm$0.161 & 0.741$\pm$0.199 & 0.809$\pm$0.133 \\
& IM              & 0.366$\pm$0.131 & 0.497$\pm$0.157 & 0.649$\pm$0.157 & 0.677$\pm$0.188 & 0.751$\pm$0.185 & 0.879$\pm$0.132 \\
& Random          & 0.518$\pm$0.110 & 0.861$\pm$0.185 & 1.120$\pm$0.247 & 1.330$\pm$0.281 & 1.619$\pm$0.303 & 2.098$\pm$0.417 \\
\bottomrule
\end{tabular}
\end{table*}

\subsection{Case Study: Which Sources Does CauMax Select?}
\label{sec:case}
To understand why the causal methods win, we examine the subsets selected on BC at $K=20$, shown in Table~\ref{tab:case}. Degree picks the most-connected sources but only median-effectiveness ones, with average $a_i{=}1.2$, so their joint effect is small at $\mathrm{Co2G}=0.21$. CauMax-G instead targets high-effectiveness sources of average $a_i{=}3.5$ that are not the most connected. It overlaps far more with the oracle subset, with Jaccard $0.33$ against Degree's $0.11$, and reaches $\mathrm{Co2G}=0.49$, $2.3\times$ Degree. As Table~\ref{tab:case} shows, across all methods, selected-source effectiveness tracks Co2G almost monotonically. These influential but not most-connected sources are what a degree heuristic misses and a causal estimand identifies.

\begin{table}[t]
\caption{Case study on BC at $K{=}20$ for a single representative draw. Jaccard is the overlap with the oracle subset.}
\label{tab:case}
\centering\scriptsize
\setlength{\tabcolsep}{3pt}
\renewcommand{\arraystretch}{0.9}
\begin{tabular}{lccccc}
\toprule
\textbf{Method} & \textbf{Co2G}\,$\uparrow$ & \textbf{Regret}\,$\downarrow$ & \textbf{Jaccard}\,$\uparrow$ & \textbf{avg.\ deg.} & \textbf{avg.\ $a_i$} \\
\midrule
Oracle          & 0.579 & 0.000 & 1.00 & 118.4 & 4.30 \\
\cmidrule(lr){1-6}
CauMax-G        & 0.486 & 0.093 & 0.33 & 120.5 & 3.54 \\
CauMax-D        & 0.492 & 0.087 & 0.33 & 147.1 & 3.33 \\
NetEst-G        & 0.333 & 0.246 & 0.21 & 149.3 & 2.19 \\
\cmidrule(lr){1-6}
IM              & 0.306 & 0.273 & 0.21 & 160.2 & 1.88 \\
Weighted-Degree & 0.240 & 0.339 & 0.14 & 158.8 & 1.52 \\
Degree          & 0.213 & 0.366 & 0.11 & 164.6 & 1.24 \\
Random          & 0.091 & 0.488 & 0.00 &  73.8 & 1.42 \\
\bottomrule
\end{tabular}
\end{table}

\subsection{Sensitivity to the Uncertainty Penalty $\lambda$}
\label{sec:lambda}
Varying $\lambda \in \{0, 0.25, 0.5, 1.0, 2.0\}$ (\textbf{Q3}), both proposed methods are robust. The average regret over budgets changes by at most $\approx0.02$ across the range and stays well below every structural baseline at all $\lambda$, as reported in Table~\ref{tab:lambda}. Both methods are essentially flat across $\lambda$, with CauMax-D nearly constant and CauMax-G marginally lowest at $\lambda{=}0.5$ on both datasets. Very large penalties hurt performance slightly by favoring low-variance, low-effect subsets. The uncertainty term is thus a mild, stabilizing regularizer, and we use $\lambda=0.5$ throughout.

\begin{table}[t]
\caption{Sensitivity to the uncertainty penalty $\lambda$, reporting average Regret@$K$ over budgets. Baselines carry no uncertainty penalty and are $\lambda$-independent, so their value is listed under $\lambda{=}0.5$ and the remaining columns are inapplicable.}
\label{tab:lambda}
\centering\scriptsize
\setlength{\tabcolsep}{3pt}
\renewcommand{\arraystretch}{0.8}
\begin{tabular}{ll ccccc}
\toprule
& & \multicolumn{5}{c}{\textbf{Uncertainty penalty} $\lambda$}\\
\cmidrule(lr){3-7}
& \textbf{Method} & $0$ & $0.25$ & $0.5$ & $1$ & $2$ \\
\midrule
\multirow{7}{*}{BC}
& CauMax-G        & 0.079 & 0.073 & \textbf{0.072} & 0.086 & 0.077 \\
& CauMax-D        & 0.098 & 0.102 & 0.106 & 0.102 & 0.102 \\
\cmidrule(l){2-7}
& NetEst-G        & \na& \na& 0.196 & \na& \na\\
& IM              & \na& \na& 0.357 & \na& \na\\
& Weighted-Degree & \na& \na& 0.382 & \na& \na\\
& Degree          & \na& \na& 0.408 & \na& \na\\
& Random          & \na& \na& 0.495 & \na& \na\\
\midrule
\multirow{7}{*}{Flickr}
& CauMax-G        & 0.302 & 0.302 & \textbf{0.299} & 0.302 & 0.303 \\
& CauMax-D        & 0.355 & 0.357 & 0.359 & 0.366 & 0.355 \\
\cmidrule(l){2-7}
& NetEst-G        & \na& \na& 0.692 & \na& \na\\
& Degree          & \na& \na& 0.620 & \na& \na\\
& Weighted-Degree & \na& \na& 0.622 & \na& \na\\
& IM              & \na& \na& 0.637 & \na& \na\\
& Random          & \na& \na& 1.258 & \na& \na\\
\bottomrule
\end{tabular}
\end{table}

\subsection{Sensitivity to Unobserved Confounding}
\label{sec:confound}
To assess robustness (\textbf{Q4}) when Assumption~\ref{as:unconf} is violated, we add a latent confounder $u_i\sim\mathcal{N}(0,1)$ per source, unobserved and not mediated by $(X,G)$, that enters both the treatment propensity and the source effectiveness with strength $\gamma_u$: $\pi_i = \sigma(w_\pi^\top X_i + \gamma_u u_i + \epsilon_i)$ and $a_i = \phi(c\,\bar z_i + \gamma_u u_i)$, so treated sources tend to be more effective for reasons the estimator cannot observe, with $\epsilon_i\sim\mathcal{N}(0,0.1^2)$ and $w_\pi, c, \bar z_i$ as in Section~\ref{sec:setup}. The true Co2G uses the full $a_i$ while the estimator sees only $X$, and $\gamma_u{=}0$ recovers the unconfounded setting of the main experiments. As $\gamma_u$ grows $0\!\to\!2$, the estimator's Co2G error rises and every method's regret increases, yet both proposed methods stay ahead of every structural heuristic throughout, as shown in Fig.~\ref{fig:confound}. Even at the strongest confounding $\gamma_u{=}2$, they remain below the best degree heuristic on both datasets.

\begin{figure}[tb]
\centerline{\includegraphics[width=\figsc\columnwidth]{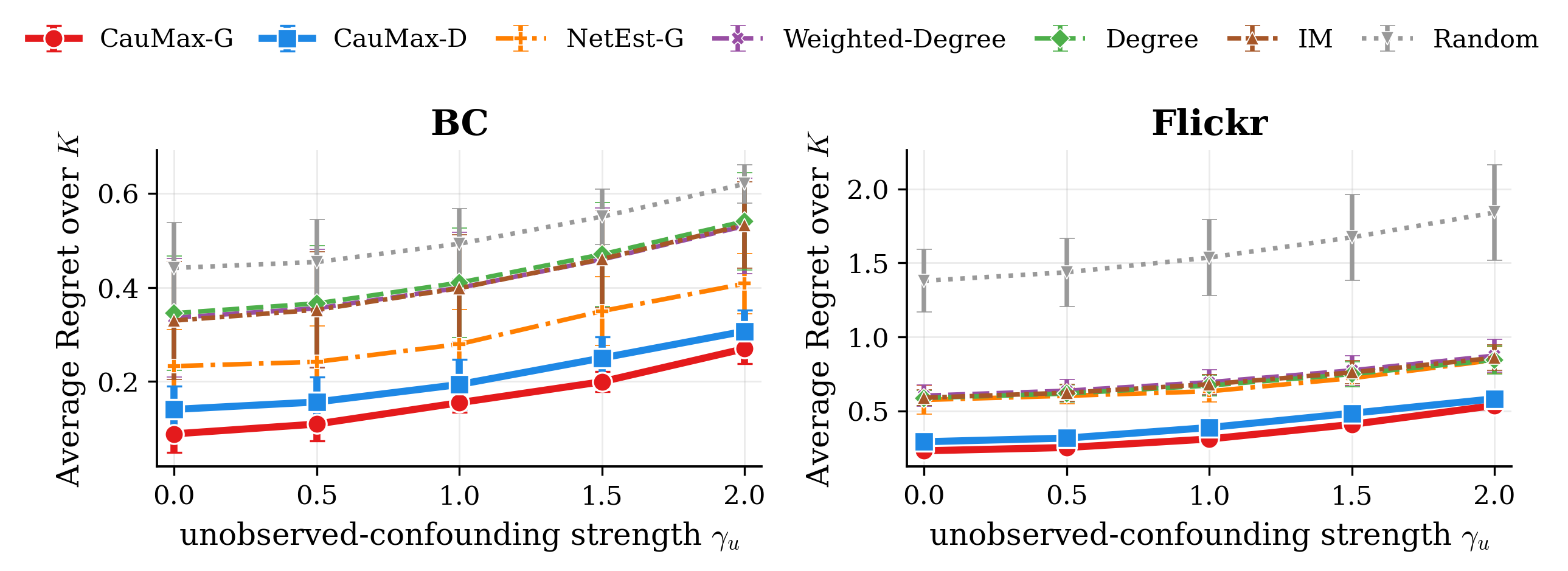}}
\caption{Sensitivity to unobserved confounding of strength $\gamma_u$.}
\label{fig:confound}
\end{figure}

\subsection{Sensitivity to the Source/Target Partition}
\label{sec:partition}
The source group is the top $p_A\%$ of nodes by coreness, with $p_A{=}15$ by default. As shown in Fig.~\ref{fig:partition}, sweeping $p_A\in\{5,10,15,20,25,30\}$ and regenerating the split, SCM, and estimator at each value, the causal methods stay well below the structural heuristics throughout and the ranking is stable across $p_A$, so the result reflects the method rather than a tuned partition.

\begin{figure}[tb]
\centerline{\includegraphics[width=\figsc\columnwidth]{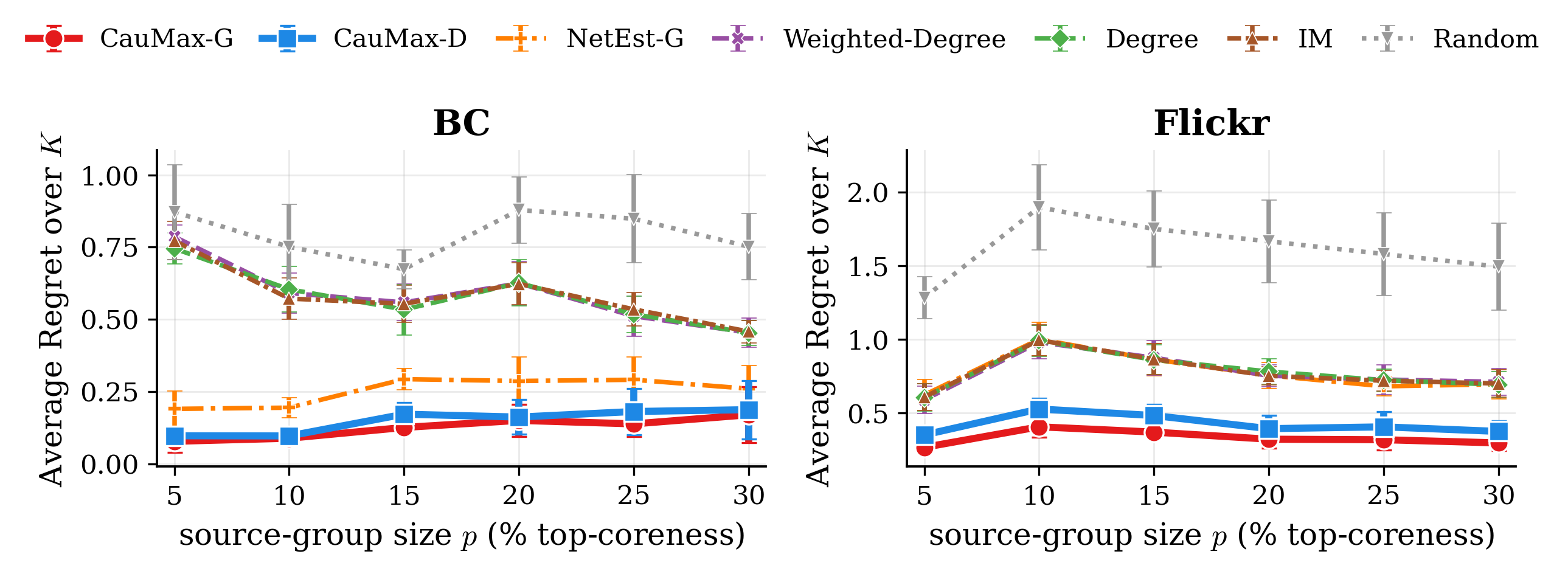}}
\caption{Sensitivity to the source-group size $p_A$, the top $p_A$\% by coreness.}
\label{fig:partition}
\end{figure}

\subsection{Sensitivity to Degree-Informativeness}
\label{sec:rho}
Because our SCM deliberately makes degree a weak predictor of effectiveness, we characterize when a causal estimand is needed through a knob $\rho\in[0,1]$ that blends the standardized feature score $z_{\mathrm{feat}}=\bar z_i$ with the standardized cross-group degree $z_{\mathrm{deg}}$, $z=\sqrt{1-\rho^2}\,z_{\mathrm{feat}}+\rho\,z_{\mathrm{deg}}$ in place of $\bar z_i$, so the Spearman correlation between true effectiveness and degree rises from $\approx0$ at our main regime $\rho{=}0$ to $1$ at $\rho{=}1$. Regenerating the split, SCM, and estimator at each value over three draws, CauMax dominates across the entire range in which effectiveness is not perfectly degree-aligned, outperforming Degree for correlations up to $\approx0.74$ on both datasets, by up to ${\approx}5\times$ on BC and ${\approx}2.3\times$ on Flickr, and only at the degenerate endpoint $\rho{=}1$ does Degree overtake it. The advantage thus holds whenever influence is at least partly feature-driven, as Fig.~\ref{fig:rho} confirms.

\begin{figure}[tb]
\centerline{\includegraphics[width=\figsc\columnwidth]{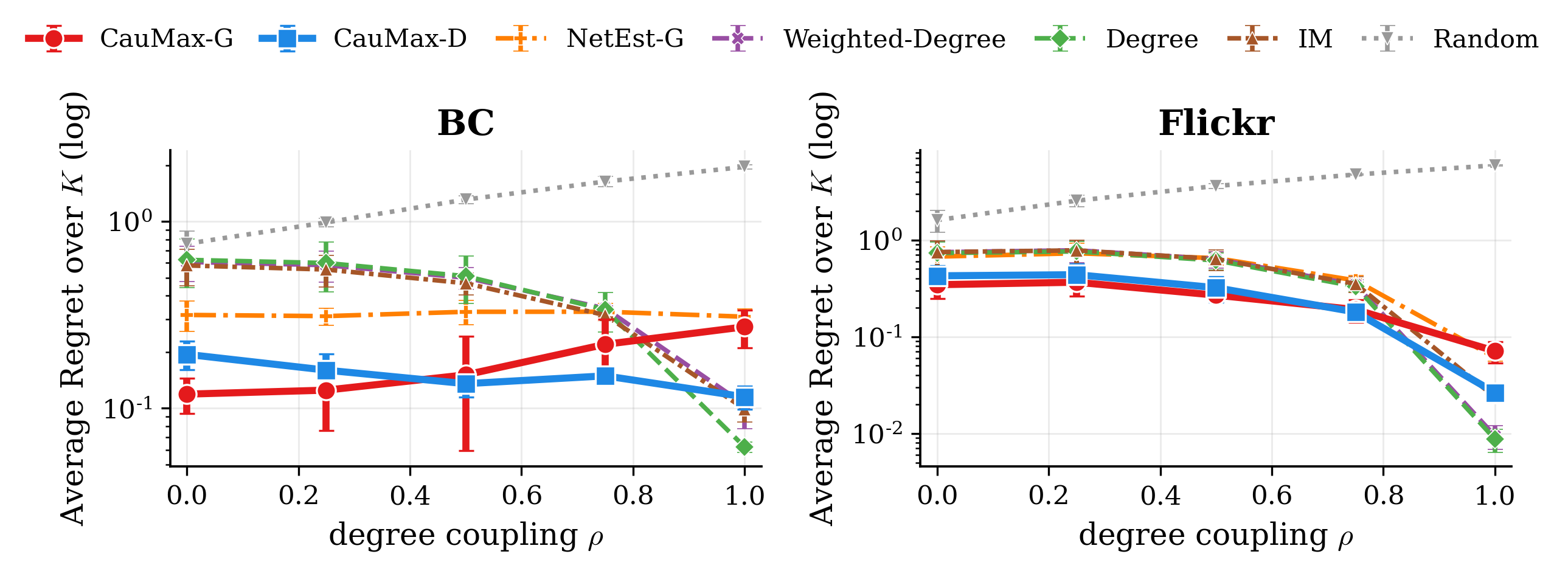}}
\caption{Sensitivity to degree-informativeness $\rho$, the Spearman correlation of true effectiveness with cross-group degree $\approx0,0.3,0.5,0.74,1.0$. CauMax dominates until $\rho{\to}1$.}
\label{fig:rho}
\end{figure}

\subsection{Robustness to SCM Design Choices}
\label{sec:misspec}
Two checks, both shown in Fig.~\ref{fig:misspec}, confirm the gains are not an SCM artifact. First, varying the effectiveness scale $c\in\{1,2,3,5\}$ leaves the ranking unchanged with CauMax-G best at every setting. Second, we misspecify the estimator by changing the data-generating channel from the $\sqrt{D_j}$ weighted sum the GNN assumes to a mean-pool or unnormalized channel, recomputing the oracle each time, and CauMax-G keeps the lowest regret under every channel. This is most telling under mean-pool, where the channel matches NetEst-G's mean-pooled exposure, so NetEst-G is correctly specified and CauMax is not, yet CauMax-G still wins. The advantage comes from the feature-aware effectiveness, greedy search, and uncertainty penalty, not from matching the channel.

\begin{figure}[t]
\centerline{\includegraphics[width=\figsc\columnwidth]{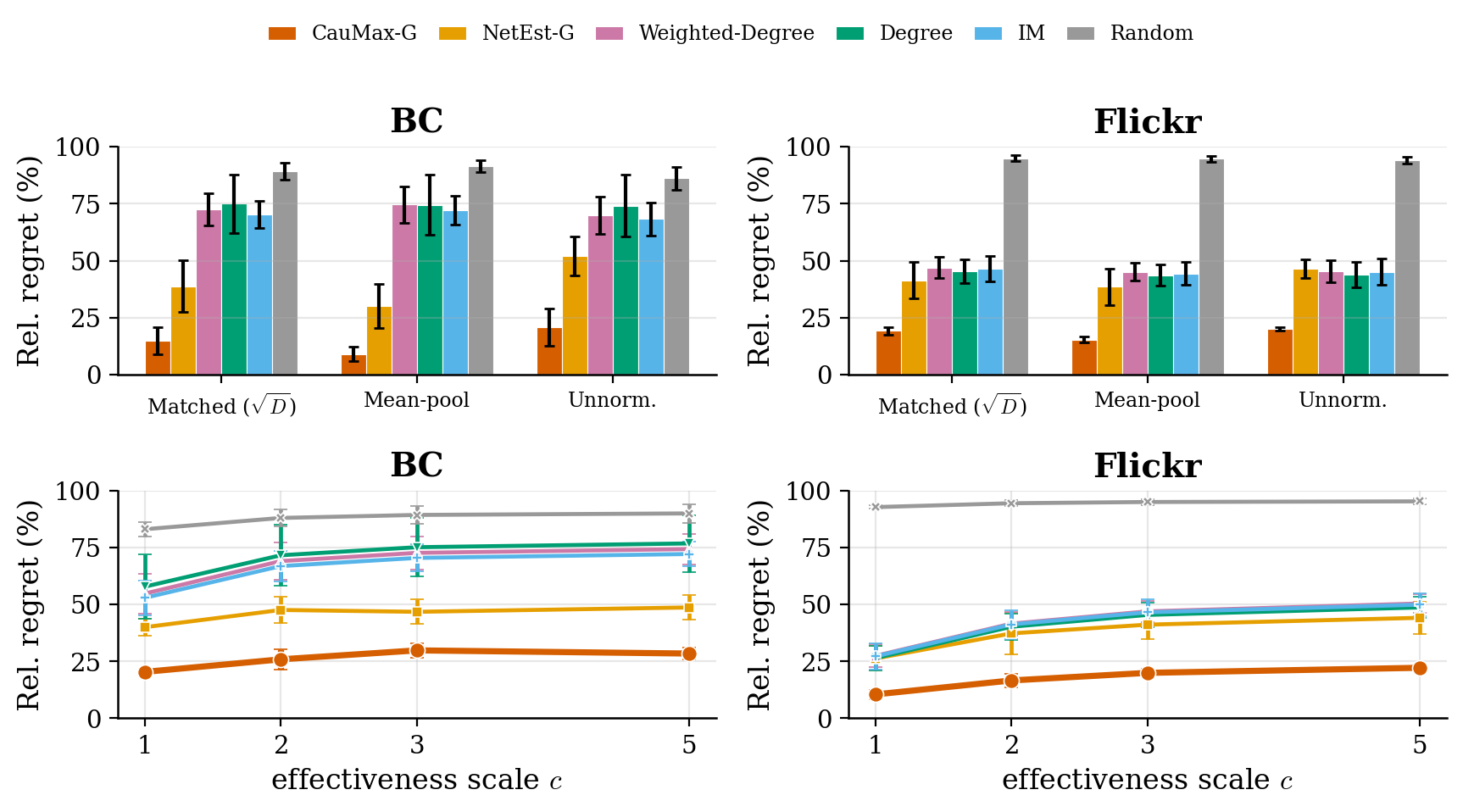}}
\caption{Robustness of the method ranking. \textbf{Top}: relative regret (\% of the oracle gain lost) under the matched $\sqrt{D_j}$ exposure versus the misspecified mean-pool and unnormalized channels. \textbf{Bottom}: relative regret across the effectiveness scale $c$.}
\label{fig:misspec}
\end{figure}

\subsection{Estimation-Error Convergence}
\label{sec:converge}
Lemma~\ref{lem:bound} bounds the Co2G error by twice the per-target outcome error $\varepsilon_n$, whose generalization term decays as $O(n^{-1/2})$. Varying the number of training target nodes over $n\in\{25,\dots,1600\}$, $\varepsilon_n$ decreases monotonically and is well captured by a floor$+\kappa\,n^{-1/2}$ fit, as plotted in Fig.~\ref{fig:converge}, consistent with the $O(n^{-1/2})$ generalization term above an irreducible floor, so more data tightens the estimated Co2G.

\begin{figure}[t]
\centerline{\includegraphics[width=\figsc\dimexpr0.70\columnwidth\relax]{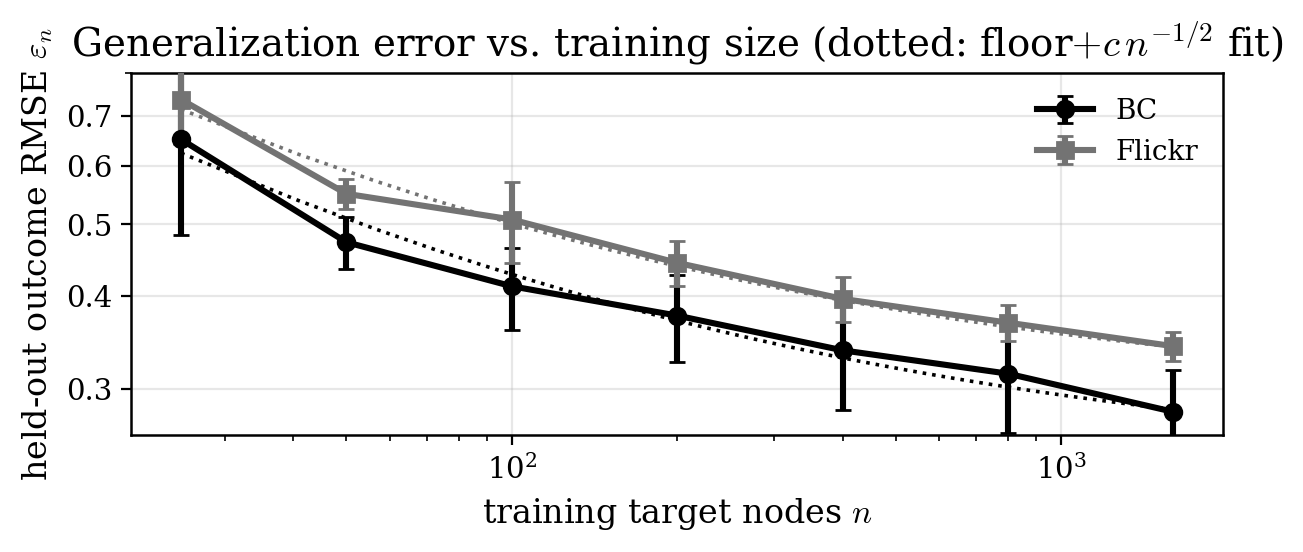}}
\caption{Held-out outcome RMSE $\varepsilon_n$ vs.\ training size $n$ (log-log). The dotted line is the floor${+}\kappa\,n^{-1/2}$ fit.}
\label{fig:converge}
\end{figure}

\subsection{Validation on Real Outcomes}
\label{sec:real}
Ground-truth Co2G is unobservable on real networks, so regret is measured on semi-synthetic data. We additionally validate (\textbf{Q5}) on real outcomes from three networks~\cite{ref_musae}, the GitHub developer graph and the Twitch-ES and Twitch-ENGB streamer graphs, splitting each into a source group and a target group, summarized in Table~\ref{tab:real}. On GitHub the estimator predicts held-out outcomes accurately, with area under the ROC curve (AUC) $0.89$. A placebo permuting the source treatments collapses the estimated cross-group effect by roughly $4\times$ on GitHub and $6\times$ on Twitch-ES, confirming a genuine signal, while the negative control Twitch-ENGB, which has little cross-group structure, yields a near-zero effect the placebo does not reduce.
\begin{table}[t]
\caption{Validation on real outcomes: held-out AUC/accuracy and the estimated cross-group effect under the real treatment vs.\ a permuted-treatment placebo. Effect columns use the subset selected by CauMax-G, while AUC and accuracy are the shared estimator's held-out predictions.}
\label{tab:real}
\centering\footnotesize
\setlength{\tabcolsep}{3.5pt}
\begin{tabular}{lccccc}
\toprule
\textbf{Network} & $|V_A|/|V_B|$ & \textbf{AUC} & \textbf{Acc.} & \textbf{eff.\ real} & \textbf{eff.\ plac.} \\
\midrule
GitHub    & 5.7K/28.6K & 0.894 & 0.857 & 0.030 & 0.008 \\
Twitch-ES & 0.7K/3.6K  & 0.654 & 0.731 & 0.044 & 0.007 \\
Twitch-ENGB & 1.1K/4.9K & 0.649 & 0.613 & 0.014 & 0.016 \\
\bottomrule
\end{tabular}
\end{table}

\section{Conclusion}
We introduced cross-group causal influence maximization: selecting the source subset whose intervention most improves a distinct target group under interference. We formalized the core-to-group causal effect (Co2G) and developed an identification theory establishing its nonparametric identifiability from observational network data via do-calculus. We proposed CauMax, a structured GNN estimator of Co2G with an uncertainty-aware objective realized as a greedy search (CauMax-G) and a differentiable optimizer (CauMax-D). Across five real-world networks, two with semi-synthetic ground truth and three with real outcomes, CauMax reduces selection regret over structural, diffusion-based, and representation-learning baselines by wide margins and stays ahead under moderate confounding. A promising direction for future work is handling temporally evolving networks with budget-adaptive sequential selection.

\section*{Acknowledgment}
The work has been partially supported by the ARC Discovery Project (DP230101122).

\balance

\end{document}